\documentclass[twocolumn]{autart}
\usepackage{amsmath,amssymb}
\usepackage{color}

\usepackage{natbib}

\usepackage{amsmath}
\usepackage{graphicx}
\usepackage{subfigure}
\usepackage{booktabs}
\usepackage{url}
\usepackage{multirow}
\usepackage{mathrsfs}
\usepackage{diagbox}
\usepackage{hyperref}

\usepackage{algorithm}
\newtheorem{remark}{Remark}
\newtheorem{theorem}{Theorem}
\newtheorem{assumption}{Assumption}
\newtheorem{lemma}{Lemma}

\hypersetup{
	colorlinks = true,
	citecolor = cyan,
}

\newenvironment{proof}{{\indent \indent \bf Proof.}}{\hfill $\qed$\par}
\newcommand{\vecc}{\rm {vec}}

\newcommand{\beq}{\begin{equation}\begin{aligned}}
		\newcommand{\eeq}{\end{aligned}\end{equation}}
\newcommand{\beqn}{\begin{equation*}\begin{aligned}}
		\newcommand{\eeqn}{\end{aligned}\end{equation*}}
\begin{document}
\begin{frontmatter}

\title{Distributed Least Squares Algorithm for Continuous-time Stochastic Systems Under Cooperative Excitation Condition\thanksref{footnoteinfo}} 
 
\thanks[footnoteinfo]{This work was supported by the National Key R\&D Program of China under Grant 2018YFA0703800, the Strategic Priority Research Program of Chinese Academy of Sciences under Grant No. XDA27000000, Natural Science Foundation of China under Grant U21B6001, and National Science Foundation of Shandong Province (ZR2020ZD26).\\Corresponding author: Zhixin Liu.}

\author[AMSS,CAS]{Xinghua Zhu}\ead{zxh@amss.ac.cn},
\author[AMSS,CAS]{Zhixin Liu}\ead{lzx@amss.ac.cn}

\address[AMSS]{Key Laboratory of Systems and Control, Institute of Systems Science, Academy of Mathematics and Systems Science,\\
	Chinese Academy of Sciences, Beijing 100190, P. R. China.}
\address[CAS]{School of Mathematical Sciences, University of Chinese Academy of Sciences, Beijing 100049, P. R. China.}


\begin{abstract}                
In this paper, we study the distributed adaptive estimation problem of continuous-time stochastic dynamic systems over sensor networks where each agent can only communicate with its local neighbors. A distributed least squares (LS) algorithm based on diffusion strategy is proposed such that the sensors can cooperatively estimate the unknown time-invariant parameter vector from continuous-time noisy signals. By using the martingal estimation theory and Ito formula, we provide upper bounds for the estimation error of the proposed distributed LS algorithm, and further obtain the convergence results under a cooperative excitation condition. Compared with the existing results, our results are established without using the boundedness or persistent excitation (PE) conditions of regression signals.
We provide simulation examples to show that multiple sensors can cooperatively accomplish the estimation task even if any individual can not.
\end{abstract}

\end{frontmatter}

\section{Introduction}
\label{sec:introduction}
In recent years, the distributed parameter estimation problem attracts much attention of researchers in diverse fields such as adaptive control, statistical learning, and signal processing (cf., \cite{Davies1970System,1982System}). In comparison with centralized algorithms (cf., \cite{pao1994centralized}) where a fusion center is required to collect and process data from all sensors, each sensor in distributed ones (cf., \cite{2001The,1994Distributed}) is used to estimate the unknown parameters by local information exchange from its neighbors. The distributed algorithms have the advantages of robustness to the node failure, reduction of the communication cost and the calculation burden. Thus, they are widely applied in many engineering systems such as target localization, intrusion detection and cooperative spectrum sensing (cf.,\cite{sayed2013diffusion,wang2014distributed}).

The least squares (LS) algorithm is a classical and elegant algorithm in system identification, statistics and many other fields  due to its fast convergence speed and simple calculation, and there are many studies on the performance analysis of the algorithm (cf., \cite{karami2007tracking,guo1995convergence,lai1982least}). With the development of computer science and communication technology, distributed LS algorithms have been widely studied where diffusion (cf., \cite{cattivelli2009diffusion,xie2018analysis}) and consensus (cf., \cite{carli2008distributed,schizas2009distributed}) are two commonly used strategies in the design of distributed algorithms. For discrete-time systems, \cite{bertrand2011diffusion} considered the diffusion-based bias-compensated distributed 
LS algorithm to estimate time-invariant parameters, and established the mean-square stability of the proposed algorithm 
where the regressors are required to satisfy the independency and stationarity assumptions.  \cite{mateos2012distributed} analyzed the stability and steady-state mean-square error performance of the consensus-type LS algorithm with independent spatio-temporally white regressors, and \cite{yu2019robust} analyzed the mean-square convergence of a robust diffusion LS algorithm with impulsive noise under the ergodic assumption on the regressors. Note that in almost all results on the performance analysis of distributed LS algorithms, regression vectors are required to satisfy the independency, stationarity or ergodicity assumptions which are hard to be applied to feedback control systems. To overcome this issue, \cite{xie2020convergence} proposed a distributed least squares algorithm by diffusing local estimation and inverse of  information matrix, and established the convergence of the algorithm under cooperative excitation condition which is a generalization of the weakest possible condition for LS in \cite{lai1982least}.


As we know, (stochastic) differential equations are often used to describe the dynamical behavior of natural and engineering systems and the coefficients of the differential equation have practical physical meaning (cf., \cite{feixianxing,sobczyk2001stochastic}). The distributed adaptive estimation problem for continuous-time systems with deterministic regression vectors has also been studied.
For example, \cite{6578135} studied the consensus-type distributed identification algorithm for continuous-time systems where  regressors are uniformly bounded  and satisfy the cooperative persistent excitation (PE) condition. \cite{2012Continuouss} investigated the exponential stability of diffusion-type LMS algorithm with PE regressors. \cite{papusha2014collaborative} investigated the asymptotic parameter convergence  of consensus-type distributed gradient algorithm under PE condition when the topology is undirected. \cite{javed2021excitation} analyzed the uniform exponential stability of the estimation error for cooperative gradient algorithm with cooperative PE regressors when the topology is directed. In the presence of noise in the model, \cite{2012Continuouss} and \cite{inproceedings} illustrated that continuous-time distributed LMS algorithm with PE regressors will remain stable when there is sufficiently small bounded noise in the system.  \cite{zhou2013distributed} analyzed the
stability of the consensus-type estimation algorithm under an observability condition. In the performance analysis of distributed adaptive estimation algorithms of continuous-time systems, the regressors are required to be deterministic and satisfy the PE condition or a condition similar to PE.  How to relax these conditions and establish the theoretical analysis of the distributed adaptive estimation algorithms of continuous-time stochastic systems with stochastic regression vectors remains unresolved.


In this paper, we propose a diffusion-type LS algorithm of continuous-time systems described by stochastic integral equations. To deal with challenges caused by the correlated noise and the stochastic differential equations, we use the continuous-time martingale convergence theory (c.f., \cite{christopeit1986quasi}) and the ITO integral (c.f., \cite{P1980Statistics}). Based on this, we establish the convergence result of the proposed algorithm under a cooperative excitation condition on stochastic regressors. The main contributions are summarized as follows.

\begin{enumerate}
	\item We present a distributed LS algorithm to estimate an unknown parameter vector of continuous-time stochastic systems according to two steps:  the adaption of the
	continuous-time innovation process and the diffusion of the local estimation and the inverse of information matrix from neighboring sensors.
	\item We introduce a cooperative excitation condition on stochastic regression signals under which the convergence result of the algorithm can be established. The cooperative excitation condition can be degenerated to the weakest condition for the continuous-time LS algorithm for single agent (cf. \cite{2006Continuous}). We remark that the cooperative excitation condition is  more genreal than the PE conditions commonly used in the existing literature.
	\item We provide simulation examples to illustrate the cooperative effect of multiple sensors in the sense that multiple agents can cooperatively accomplish the estimation task even if any individual can not.
\end{enumerate}

The rest of this paper is arranged as follows. The problem formulation including some preliminaries and the distributed LS algorithm is introduced in Section \ref{PROBLEM FORMULATION}.  The convergence of the proposed algorithm is established in Section \ref{Convergence}. A numerical simulation is given in Section \ref{SIMULATION}. The concluding remarks are made in Section \ref{conclusion}.
\section{Problem Formulation}\label{PROBLEM FORMULATION}
\subsection{Some Preliminaries}\label{Preiminary}
\subsubsection{Matrix Theory}
In this paper, we use $A=(a_{ij})\in \mathbb{R}^{m\times n}$ to denote an $m\times n$-dimensional real matrix, and $I_m$ to denote the $m$-dimensional identity matrix. For an $m\times m$-dimensional matrix $A$, $\lambda_{\max}(A)$ and $\lambda_{\min}(A)$ denote the maximum and  minimum eigenvalues of $A$, and $tr(A)$ and $|A|$ denote the trace and determinant of $A$, respectively.  $\|A\|$ represents the Euclidean norm, i.e., $\|A\|\overset{\Delta}{=}(\lambda_{\max}\{AA^T\})^{\frac{1}{2}}$, where the notation $T$ denotes the transpose of the matrix. The matrix $A$ is called stochastic if all the row sums of $A$ equal to 1. Furthermore, the matrix $A$ is called doubly stochastic if both the row sums and the column sums equal to 1. For a matrix sequence $\{X_k, k\geq 0\}$ and a positive scalar sequence $\{b_k, k\geq 0\}$,  $X_k=O(b_k)$
means that there exists a positive constant $M$ independent of $k$, such that $\|X_k\|\leq Mb_k$ holds for all $k\geq 0$, and  $X_k=o(b_k)$ means that $\lim\limits_{k\rightarrow\infty} \frac{\|X_k\|}{b_k} =0$. We use $\log(\cdot)$ to denote the natural logarithmic function. The Kroneker product $A\otimes B$ of two matrices $A=(a_{ij})\in \mathbb{R}^{m\times n}$ and $B=(b_{ij})\in\mathbb{R}^{p\times q}$ is defined as
\begin{center}$
	A\otimes B=\left[
	\begin{matrix}
		a_{11}B &\cdots &a_{1n}B\\
		\vdots    &\ddots & \vdots\\
		a_{n1}B &\cdots &a_{nn}B\\
	\end{matrix}
	\right] \in \mathbb{R}^{mp\times nq}.
	$\\
\end{center}


\subsubsection{Graph Theory}
The communication between sensors (or agents) are modelled as an undirected graph $ \mathcal G = (\mathcal{V, E})$,  where $\mathcal V=\{1,2,\cdots, N\}$ is composed of all sensors and $\mathcal E\subseteq \mathcal V\times \mathcal V$ is the set of edges. The weighted adjacency matrix $\mathcal{A}=(a_{ij})_{n\times n}$ is used to describe the interaction weights between sensors, where the weight $a_{ij}\textgreater 0$ if and only if $(i,j)\in \mathcal{E}$. For simplicity of analysis, we assume that the matrix $\mathcal{A}$ is symmetric and doubly stochastic. A path of length $s$ in the graph $\mathcal G$ is defined as a sequence of labels of sensors $\{i_1,\cdots,i_s\}$ satisfying $(i_j,i_{j+1}) \in \mathcal{E}$ for all $1\leq j\leq s-1$. The diameter of the graph $\mathcal{G}$, denoted as $D_{\mathcal{G}}$, is defined as the maximum length of the path between any two sensors. We define $N_i=\{j\in \mathcal{V}|(j,i)\in\mathcal{E}\}$ as the neighbor set of the graph $\mathcal{G}$. See \cite{[30]} for more information about the graph theory.

\subsubsection{Continuous-Time Martingale}
Let $(\Omega,\mathcal{F},P)$  be a probability space, and $\{\mathcal{F}_t,t\geq0\}$ be a nondecreasing family of sub-$\sigma$-algebras of $\mathcal{F}$.
The process $\{X_t,\mathcal{F}_t;0\leq t < \infty\}$ is said to be a martingale if we have $E(X_t|\mathcal{F}_s)=X_s$ almost surely (a.s.) for  $0\leq s < t< \infty$ , where $E[\cdot|\cdot]$ is the conditional
mathematical expectation operator. The Wiener process $\{w(t),\mathcal{F}_t\}$ is an independent incremental process and a square-integrable martingale, i.e., $E[w(t)|\mathcal{F}_s]=0$ and $E[w^2(t)]<\infty, t>s\geq0,$ where $E[\cdot]$ is the
mathematical expectation operator.
For the continuous-time martingale, the following martingale estimation theorem is often used to deal with the continuous-time stochastic noise.

\begin{lemma}\label{lA.2}\citep{chen1991identification} Let $(M_t, \mathcal{F}_t)$ be a measurable process satisfying
	$\int_0^t\|M_s\|^2ds<\infty$,\ \ \ a.s.\ \ $\forall t>0$.
	If $\{ w_t, \mathcal F_t\}$ is a Wiener process, then as $t \rightarrow\infty$,
	$\int_0^t M_s dw_s=O\bigg(\sqrt	{S(t) \log \log(S(t)+e)}\bigg) \ \ \ a.s.,$ 
	where $S(t)$ is defined by $S(t)=\int_0^t\|M_s\|^2ds.$
\end{lemma}

\subsubsection{Stochastic Differential Equation}

For the Borel-measurable functions $a(t, x): [0,\infty)\times\mathbb{R}^d\to\mathbb{R}$ and $b(t, x): [0,\infty)\times\mathbb{R}^d\to\mathbb{R}$, the (a. s. continuous) random process $\xi=\{\xi(t), t>0\}$ is said to be a strong solution of the stochastic differential equation $d\xi(t)=a(t, \xi(t))dt+b(t, \xi(t))dw(t)$ with the $\mathcal{F}_0$-measurable initial condition $\xi(0)$ if
\begin{enumerate} 
	\item for each $t>0$, the random variable $\xi(t)$ is $\mathcal{F}_t$-measurable.
	\item
	$|a(\cdot,\cdot)|^{\frac{1}{2}}\in \mathcal{P}_t$,
	$|b(\cdot,\cdot)|\in \mathcal{P}_t$, i.e.,
	$P\big(\int_0^T|a(t, \xi(t))|dt<\infty\big)=1$, $P\big(\int_0^T|b(t, \xi(t))|^2dt<\infty\big)=1.$
	\item $\xi(t)=\xi(0)+\int_0^ta(s, \xi(s)) ds+\int_0^tb(s, \xi(s)) dw(s), t>0, a.s..$
\end{enumerate}

The Ito formula plays a key role to deal with continuous-time stochastic processes, which is described as follows,
\begin{lemma}\label{lA.1}\citep{IntroductiontoControlTheory} Assume that the stochastic process $\xi(t)$ obeys the equation $d\xi(t)=a(t,\xi(t)) dt+B(t,\xi(t))
	dw(t)$, and  \{$a(t,\xi(t)),\mathcal{F}_t$\} is an $l$-dimensional adaptive process and $\{B(t,\xi(t)),\mathcal{F}_t\}$ is an $l\times m$-dimensional adaptive matrix process satisfying $\|a(\cdot,\cdot)\|^{\frac{1}{2}}\in \mathcal{P}_t$ and $\|B(\cdot,\cdot)\|\in \mathcal{P}_t$. If the functions $f_t^{'}(t,x)$, $f_{x}^{'}(t,x)$ and $f_{xx}^{''}(t,x)$ are continuous, 
	then
	\begin{align}
		&df(t,\xi(t))=f_t^{'}(t,\xi(t))dt+{f^{'}}^T_{\xi(t)}(t,\xi(t)) d\xi(t)\notag\\&\hskip 1cm+\frac{1}{2}tr\{f_{\xi(t)\xi(t)}^{''}(t,\xi(t))B(t,\xi(t)) (B(t,\xi(t)))^T\}dt.\notag
	\end{align}
\end{lemma}

\subsection{Distributed LS Algorithm}\label{Distributed LS Algorithm}
Consider a network of $N$ sensors whose dynamics obey the following continuous-time stochastic differential equations with a general form,
\begin{align} \label{2.1}
	a(S)y_{i}(t)=Sb(S)u_{i}(t)+c(S)v_{i}(t), i\in\{1,\cdots,N\},
\end{align}
where $y_i(0)=u_i(0)=0$, $S$ is the integral operator (i.e, $Sy_i(t)=\int_0^ty_i(s)ds$),  $y_{i}(t)$ and $u_{i}(t)$ are scalar output and input of the agent $i$ at time $t$, and $a(S)$, $b(S)$ and $c(S)$ are three polynomials of the integral operator $S$  with unknown coefficients $\{a_i,b_j,c_l, 1\leq i\leq p,0\leq j\leq q,1\leq l\leq r\}$,
\begin{align}
	a(S)&=1+a_1S+\cdots+a_pS^p,\ p \geq 0, \notag\\
	b(S)&=b_1+b_2S+\cdots+b_qS^{q-1},\ q \geq 1, \notag\\
	c(S)&=1+c_1S+\cdots+c_rS^r,\ r \geq 0.\notag
\end{align}
In \eqref{2.1}, the system noise $v_{i}(t)$ is generated from a standard Wiener process $\{w_{i}(t),\mathcal {F}_t\}$, i.e.,
\begin{align} \label{zs}
	d(S)v_{i}(t)= w_{i}(t),\  t \geq 0,
\end{align}
where $\{\mathcal {F}_t, t\geq0\}$ is a family of nondecreasing $\sigma$-algebras defined as  $\mathcal {F}_t\triangleq\sigma\{y_i(s), u_i(s), v_i(s),i=1,\cdots,N,0\leq s\leq t\}$, and $d(S)=1+d_1S+\cdots+d_rS^r,\ r \geq 0$ is a known stable filter.
Denote the collection of unknown coefficients in \eqref{2.1} as
$\theta=[-a_1, \cdots, -a_p,  b_1, \cdots, b_q, c_1, \cdots, c_r]^T$.
Correspondingly, denote the regressor as
\begin{align}\label{phi0}
	\phi_i^{0}(t)=&[y_{i}(t), \cdots, S^{p-1}y_{i}(t), u_{i}(t), \cdots, S^{q-1}u_{i}(t), \notag\\
	&v_{i}(t), \cdots, S^{r-1}v_{i}(t)]^T.
\end{align}
Then, the system \eqref{2.1} can be rewritten into the following linear time-invariant regression
model,
\begin{align}\label{xt} 
	y_{i}(t)=S\theta^T\phi_i^{0}(t)+v_{i}(t),\ t\geq0.
\end{align}

The dynamics of many systems in engineering practice can be written as \eqref{2.1} according to the laws of physics, such as mechanical arm system, mass-spring-damping system, RLC circuit system, etc. The estimation of unknown parameters obtained by using input and output signals of the continuous-time systems is highly interpretable and has practical physical significance. Furthermore, it is of great significance for system fault diagnosis and system operating life prediction \citep{rao2006identification}.

The purpose of this paper is to design the distributed algorithm for all sensors to cooperatively estimate the unknown parameter vector $\theta$ by utilizing the local signals $\{y_j(s),u_j(s),0\leq s\leq t\}_{j\in {N_i}}$. We know that the least squares (LS) algorithm has attracted much attention of researchers due to fast convergence rate and widely applications in engineering systems. In this paper, we put forward the distributed LS (DLS) algorithm by combining the continuous-time LS algorithm with the diffusion of local information at the discrete-time instants $t_0=0, t_1, \cdots, t_k,\cdots$, which results in a hybrid algorithm. The details of the algorithm are described in the following Algorithm \ref{algorithm1}.

\begin{algorithm}[!ht]
	\caption{DLS Algorithm}\label{algorithm1}
	For the agent $i\in\{1,\cdots,N\}$, we begin with an initial estimate $\theta_{0,i}(0)$ and an initial  positive definite matrix $P_{0,i}(0)$. The time instants for the diffusion process are denoted as $0=t_0, t_1, t_2, \cdots.$
	
	\textbf{Step 1: Adaptation.} For  $t\in [t_k,t_{k+1})$, ${\theta_{k,i}(t)}$ and ${P_{k,i}(t)}$ are generated according to the following equations,
	\begin{align}\label{00y}
		& d {\theta_{k,i}(t)}={P_{k,i}(t)}\phi_{k,i}(t)d(S)\{d{y_i}(t)-\phi_{k,i}^T(t) {\theta_{k,i}(t)}dt\}, \\
		&\label{000y} d {P_{k,i}(t)}=-{P_{k,i}(t)}\phi_{k,i}(t)\phi_{k,i}^T(t){P_{k,i}(t)}dt,
	\end{align}
	where \begin{align}\label{phi1}
		\phi_{k,i}(t)&=[y_{i}(t),\cdots,S^{p-1}y_{i}(t),u_{i}(t),\cdots,S^{q-1}u_{i}(t),\notag\\
		&\ \ \ \ \ \ \ \ \hat {v}_{k,i}(t), \cdots,S^{r-1}\hat{v}_{k,i}(t)]^T,\\ \label{vtij}
		\hat{v}_{k,i}(t)&=y_{i}(t)-S\phi_{k,i}^T(t)\theta_{k,i}(t).
	\end{align}
	
	\textbf{Step 2: Diffusion.}  At the time instant $t_{k+1}$,   $P_{k+1,i}^{-1}(t_{k+1})$ and $\theta_{k+1,i}(t_{k+1})$ are updated by the following equations,
	\begin{align} 
		\label{rh}P_{k+1,i}^{-1}(t_{k+1})&=\sum_{j \in N_i}{a_{ij}}P_{k,j}^{-1}(t_{k+1}),\\
		\label{rh1}\theta_{k+1,i}(t_{k+1})&= P_{k+1,i}(t_{k+1})\sum_{j \in N_i}a_{ij}P_{k,j}^{-1}(t_{k+1})\theta_{k,j}(t_{k+1}),
	\end{align}
	where  $P_{k,i}^{-1}(t_{k+1})$ and $\theta_{k,i}(t_{k+1})$ are obtained by \textbf{Step 1}.
\end{algorithm}

For $r=0$, the stochastic differential equations \eqref{00y}-\eqref{rh1} have a unique strong solution if for  $i\in \{1,2,\cdots,N\}$, $P\big(\int_{t_k}^{t_{k+1}}\|P_{k,i}(t)\phi_{k,i}(t)\phi_{k,i}^T(t)\|dt<\infty\big)=1$ holds for all positive integer $k$.
For the general case of $r>0$, the existence and uniqueness of the strong solution of stochastic differential equations \eqref{00y}-\eqref{rh1} become challenging, and we will not discuss this issue in this paper. To proceed the analysis of the distributed algorithm,  we assume that for $i \in\{1,\cdots,N\}$, the stochastic differential equations \eqref{00y}-\eqref{rh1} have a unique strong solution $\{\theta_{k,i}(t), t\in [t_k,t_{k+1}), k=0,1,\cdots\}$.

\begin{remark}\label{remark} For the case of $d(S)=1$, we can obtain  Algorithm \ref{algorithm1} by minimizing the following ``accumulative prediction error"  for $t\in [t_{n},t_{n+1}) (n=0,1,\cdots$),
	\begin{align}		\delta_{t,i}(\theta)=&\sum_{j=1}^N\sum_{k=0}^{n-1}a_{ij}^{(n-k+1)}\int_{t_k}^{t_{k+1}}\big\{\frac{dy_j(s)}{ds}-\theta^T\phi_{k,j}(s)\big\}^2ds\notag\\&+\int_{t_{n}}^t\big\{\frac{dy_i(s)}{ds}-\theta^T\phi_{n+1,i}(s)\big\}^2ds.\notag
	\end{align}
	%
\end{remark}
\section{Performance Analysis of Algorithm \ref{algorithm1}.}\label{Convergence}
For convenience of analysis of Algorithm \ref{algorithm1}, we introduce some notations in Table \ref{biao1},
where col$\{\cdots\}$ denotes the vector stacked by the specified vectors, $\rm{diag}\{\cdots\}$  denotes the block diagonal matrix with each block being the corresponding vectors or matrices and $\bigotimes$ represents the Kronecker product.

By the notations in Table \ref{biao1}, the continuous-time dynamical systems \eqref{zs} and \eqref{xt} can be written into the following matrix form,
\begin{align}
	\label{ytw}&Y(t)=S(\Phi^0(t))^T\Theta+V(t),\\ \label{ytw1}&d(S)V(t)=W(t).\end{align} 	
For $t\in[t_k,t_{k+1})$, (\ref{00y}) and (\ref{000y}) in Algorithm \ref{algorithm1} can be written into the following equations,
\begin{align}\label{sf1}
	\left\{
	\begin{array}{ll}
		d\Theta_k(t)=P_k(t)\Phi_k(t)d(S)(dY(t)-\Phi_k^T(t)\Theta_k(t)dt)\\
		dP_k(t)=- P_k(t)\Phi_k(t)\Phi_k^T(t)P_k(t)dt
	\end{array}
	\right.,
\end{align}
and (\ref{rh}) and (\ref{rh1}) can be written as
\begin{equation}\label{sf}
	\left\{
	\begin{array}{ll}
		{\vecc} \{P_{k+1}^{-1}(t_{k+1})\}=\mathscr{A}{\vecc}\{P_{k}^{-1}(t_{k+1})\}\\
		\Theta_{k+1}(t_{k+1})=P_{k+1}(t_{k+1})\mathscr{A}P_{k}^{-1}(t_{k+1})\Theta_{k}(t_{k+1})
	\end{array}
	\right.,
\end{equation}
where ${\rm {vec}}\{\cdot\}$ represents the operator that stacks the block matrices on top of each other.
\begin{table}\label{tab1}
	\caption{Some Notations}
	\label{biao1}
	\begin{center}
		\setlength{\tabcolsep}{3pt}
		\begin{tabular}{|c|c|}
			\hline
			Notation & Definition \\
			\hline
			$Y(t)$     & $(\it{y_{\rm1}(t),\cdots,y_N(t)})^{T}$  \\
			$\Phi^0(t)$&$\rm{diag}\{\phi_1^{0}(\it t),\cdots,\phi_N^{\rm{0}}(\it t)\}$\\
			$V(t)$     &$(\it{v_{\rm1}(t),\cdots,v_N(t)})^{T}$\\
			$\hat{V}_k(t)$&$(\it{\hat{v}_{k,{\rm1}}(t),\cdots,\hat{v}_{k,N}(t)})^T$\\
			$\widetilde{V}_k(t)$&$(\it{\widetilde{v}_{k,{\rm1}}(t),\cdots,\widetilde{v}_{k,N}(t)})^T,\it{\widetilde{v}_{k,i}(t)=v_{i}(t)-\hat{v}_{k,i}(t)}$ \\
			$W(t)$&$(\it{w_{\rm1}(t),\cdots,w_N(t)})^{T}$\\
			$\Theta$   & $\rm{col}\{\theta,\cdots,\theta\}$  \\
			$\Theta_k(t)$   & $\rm{col}\{\theta_{\it k,\rm1}(\it t),\cdots,\theta_{\it k,N}(\it t)\}$  \\
			$\widetilde{\Theta}_k(t)$&$\rm{col}\{\widetilde{\theta}_{\it k,\rm1}(\it t),\cdots,\widetilde{\theta}_{\it k,\it N}(\it t)\}$,$\widetilde{\theta}_{k,i}(t)=\theta-\theta_{k,i}(t)$\\
			$\Phi_k(t)$&$\rm{diag}\{\phi_{\it k,\rm1}(\it t),\cdots,\phi_{k,\it N}(t)\}$\\
			$\widetilde{\Phi}_k(t)$&$\rm{diag}\{\widetilde{\phi}_{\it k,\rm1}(\it t),\cdots,\widetilde{\phi}_{k,\it N}(\it t)\},\widetilde{\phi}_{k,i}(t)=\phi_i^{\rm{0
			}}(t)-\phi_{k,i}(t)$\\
			$P_k(t)$&$\rm{diag}\{\it{P_{k,\rm{1}}(t),\cdots,P_{k,\it N}(t)}\}$\\
			$\mathscr{A}$&$\mathcal{A}\otimes I_{(p+q+r)}$,~~$\mathcal{A}$ is the weighted adjacency matrix\\
			\hline
		\end{tabular}
	\end{center}
\end{table}

To proceed with the performance analysis of the algorithm, we need to introduce some assumptions concerning the network topology, the polynomials $c(S)$ and $d(S)$ and the regression vectors.
\begin{assumption} \label{c1}
	The graph $\mathcal G$ is undirected and connected.
\end{assumption}

In fact, the above assumption for the communication graph can be generalized to the case where $\mathcal G$ is a strongly connected balanced directed graph, and the corresponding analysis is similar to the undirected graph case. Thus, we just provide the analysis of Algorithm \ref{algorithm1} under Assumption \ref{c1}.

In the following, we introduce the assumption of the strictly positive real property concerning the polynomials, which is often utilized to deal with accumulative correlated noise in the stability analysis of adaptive control systems and convergence analysis of system identification.

\begin{assumption}\label{c2}
	The polynomial $d(S)c^{-1}(S)-\frac{1}{2}$ is strictly positive real.
\end{assumption}

\begin{remark}\label{rr3.1} By the strictly positive real property of the polynomial $d(S)c^{-1}(S)-\frac{1}{2}$, we see that for the functions $g(t)$ and $f(t)$ satisfying  $f(t)=(d(S)c^{-1}(S)-\frac{1}{2})g(t)$, there are positive constants $m_1$ and $m_2$ such that the inequality $
	\int_{0}^t g^T(\tau)\big\{f(\tau)-m_1 g(\tau)\big\}d\tau+m_2>0$ holds for all $t\geq0$ (cf., \cite{2006Continuous}).\end{remark}

%

By (\ref{xt}), we see that if all the regressor vectors $\phi_i^{0}(t)$ are equal to $0$, then the unknown parameter vector $\theta$ cannot be identified because the observed signals $y_i(t)$ do not contain any information about $\theta$. Therefore, we introduce the following excitation condition on the regression signals to estimate $\theta$. 

\begin{assumption} [Cooperative Excitation Condition]\label{c3} The regressor vectors satisfy the following condition,
	\begin{equation}
		\lim\limits_{n\to \infty} \frac{\log R(n)}{\lambda^n_{\min}} =0\ \ \  a.s.,\notag
	\end{equation} where 
 \begin{align}\label{jl}
		R(n)\triangleq&\sum_{j=1}^N\left\{\sum_{k=0}^{n}\int_{t_k}^{t_{k+1}}\|\phi_{k,j}(s)\|^2ds\right\}
		+\lambda_{\max}\left(P_0^{-1}(0)\right),
	\end{align}
	and $\lambda_{\min}^n\triangleq$
	\begin{align}\label{lambdamin}
		\lambda_{\min}\left\{\sum_{j=1}^N \sum_{k=0}^{n-D_{\mathcal{G}}}  \int_{t_k}^{t_{k+1}}\phi_{k,j}(s)\phi_{k,j}^T(s)ds+\sum_{j=1}^N P_{0,j}^{-1}\right\}.
	\end{align}
\end{assumption}
\begin{remark}\label{r3.1} For the single agent case,
	\cite{2006Continuous} proved the convergence of continuous-time LS algorithm under the following excitation condition,
	\begin{align}\label{jd}
		\lim\limits_{t\to \infty} \frac{\log \Big(e+\int_0^t\|\phi_{k,i}(s)\|^2ds\Big)}{\lambda_{\min}\Big\{P_{0,i}^{-1}+\int_0^t\phi_{k,i}(s)\phi_{k,i}^T(s)ds\Big\}} =0\ a.s..
	\end{align}
	The Cooperative Excitation Condition (Assumption  \ref{c3}) can be degenerated to \eqref{jd} when $N=1$ and $D_{\mathcal{G}}=1$.
	It is clear that Assumption \ref{c3} is much weaker than the cooperative PE condition commonly used for the convergence of parameter estimation of continuous-time systems \citep{6578135,2012Continuouss,papusha2014collaborative, javed2021excitation},
	\begin{align}
		\int_t^{t+T_0}\left[\sum_{i=1}^N\phi_{k,i}(s)\phi_{k,i}^T(s)\right]	ds\geq\alpha I_{p+q+r},\forall t\geq0,\notag
	\end{align}
	where $T_0$ and $\alpha$ are two positive constants.
\end{remark}

\begin{remark} The Cooperative Excitation Condition can reflect the joint effect of multiple sensors in a sense that even if any individual can not finish the estimation task, they can in a cooperative way. We will reveal this point
	by a simulation example given in Section \ref{SIMULATION}.
\end{remark}

Before presenting the main theorems, we first introduce some key lemmas. The lemma given below aims to deal with the accumulative correlated noise of all the individuals in the system.


\begin{lemma}\label{l2}
	Under Assumption \ref{c2}, there exist positive constants $m_1$ and $m_2$ such that the following inequality holds
	\begin{align}\label{STR} &\sum_{k=0}^{n-1}\int_{t_k}^{t_{k+1}}{g_k^T(\tau)\big(f_k(\tau)-m_1 g_k(\tau)\big)d\tau}\notag\\&+\int_{t_n}^t{g_n^T(\tau)\big(f_n(\tau)-m_1 g_n(\tau)\big)d\tau}+m_2>0,\end{align}
	where
	$g_k^T(t)=\widetilde{\Theta}_k^T(t)\Phi_k(t)$ and  $f_k(t)=\frac{1}{2}g_k(t)+\frac{c(S)-d(S)}{S}\widetilde{V}_k(t)$ with $\widetilde{V}_k(t)$  being
	defined in Table I.
	
\end{lemma}

\begin{proof} By the notations in Table \ref{biao1}, we have for $\tau\in [t_k, t_{k+1})$ ($k\geq 0$)\begin{align}\label{thw11}
	&\widetilde{\Phi}_k(\tau)=\Phi^0(\tau)-\Phi_k(\tau),\
	\widetilde{\Theta}_k(\tau)=\Theta-\Theta_k(\tau),\\ &\label{hv}\hat{V}_k(\tau)=Y(\tau)-S\Phi_k^T(\tau)\Theta_k(\tau).
\end{align}
By (\ref{ytw}) and (\ref{thw11}), the dynamics of the system have the following expression,
\begin{align}\label{wi}
	&dY(\tau)\notag\\=&\widetilde{\Phi}_k^T(\tau)\Theta d\tau+\Phi_k^T(\tau)\Theta d\tau+dV(\tau) \\
	=&\widetilde{\Phi}_k^T(\tau)\Theta d\tau+\Phi_k^T(\tau)\widetilde{\Theta}_k(\tau)+\Phi_k^T(\tau)\Theta_k(\tau)+dV(\tau).\notag
\end{align}
Differentiating  both sides of \eqref{hv} yields
\begin{align}\label{differ}
	\Phi_k^T(\tau)\Theta_k(\tau)=dY(\tau)-d\hat{V}_k(\tau).
\end{align}
Substituting (\ref{differ}) into \eqref{wi}, we can obtain 
\begin{align}\label{aa}
	&\widetilde{\Phi}_k^T(\tau)\Theta d\tau+\Phi_k^T(\tau)\widetilde{\Theta}_k(\tau)d\tau\notag\\&=d\hat{V}_k(\tau)-dV(\tau)=-d\widetilde{V}_k(\tau).
\end{align}
By the definition of $	\widetilde{\Phi}_k(t)$ and $\Theta$, we have
\begin{align}
	&\widetilde{\Phi}_k(\tau)=\rm{diag}\big\{[\rm0,\cdots,\rm0, \tilde {\it v}_{\it k,\rm1}(\it \tau),\cdots,S^{\it r-\rm1}\tilde{\it v}_{k,{\rm1}}(\it \tau)]^{\it T},\notag\\&\ \ \ \ \ \ \ \ \ \ \ \ \ \cdots, [\rm0,\cdots,\rm0, \tilde {\it v}_{\it k,\it N}(\it \tau),\cdots,S^{\it r-\rm1}\tilde{\it v}_{k,\it N}(\it \tau)]^{\it T}\big\},\notag \\
	&\Theta=\rm{col}\it{\big\{[-a_{\rm1} \cdots -a_p\ b_{\rm1} \cdots b_q\ c_{\rm1} \cdots c_r],}\notag\\& \ \ \ \ \ \ \ \ \ \ \ \ \ {\cdots, [-a_{\rm1} \cdots -a_p\ b_{\rm1}\cdots b_q\ c_{\rm1} \cdots c_r]\big\}}.\notag
\end{align}
Substituting them into  \eqref{aa}, we have
\begin{align}\label{h}
	c(S)d\widetilde{V}_k(\tau)=-\Phi_k^T(\tau)\widetilde{\Theta}_k(\tau)d\tau.
\end{align}
By Assumption \ref{c2}, the following equation
\begin{align}\label{s}
	d\widetilde{V}_k(\tau)=-c^{-1}(S)\Phi_k^T(\tau)\widetilde{\Theta}_k(\tau)d\tau
\end{align}
holds.
By \eqref{s} and the definition of $f_k(t)$, we have
\begin{align}
	f_k(\tau)=[d(S)c^{-1}(S)-\frac{1}{2}]\Phi_k^T(\tau)g_k(\tau).\nonumber
\end{align}  Under Assumption \ref{c2}, we see that there exist two positive constants $m_1$ and $m_2$ such that
\begin{align}\label{spr}
	\int_{t_k}^{t_{k+1}}{g_k^T(\tau)\big(f_k(\tau)-m_1 g_k(\tau)\big)d\tau}+m_2>0
\end{align}
holds for all positive integer $k$
and \begin{align}
	\int_{t_n}^t{g_n^T(\tau)\big(f_n(\tau)-m_1 g_n(\tau)\big)d\tau}+m_2>0.\nonumber
\end{align}
Summing both sides of the equation (\ref{spr}) from $0$ to $n-1$, we can
obtain the result of the lemma.
\end{proof}


The next lemma introduced in \cite{xie2020convergence} can be used to deal with the impact of neighbor relations on the convergence of the algorithm, and we will list it here.
\begin{lemma} \citep{xie2020convergence}\label{l4} 
	For the matrices $P_{k+1}(t_{k+1})$ and $P_{k}(t_{k+1})$ defined in \eqref{sf}, we have the following  inequality,
	\begin{align}
		&\mathscr{A}P_{k+1}(t_{k+1})\mathscr{A}\leq P_{k}(t_{k+1}),\notag\\
		&\big|P_{k}^{-1}(t_{k+1})\big|\leq\big|P_{k+1}^{-1}(t_{k+1})\big|.\notag
	\end{align}
\end{lemma}

Based on the above lemmas, we can obtain the following theorem for the estimation error  without requiring any excitation condition on the regression vector $\phi_{k,i}(t)$.
\begin{theorem}\label{th1}
	Under Assumption \ref{c2}, the estimation error $\widetilde{\Theta}_k(t)$ of Algorithm 1 satisfies the  following relationship,
	\begin{align}
		1)\ &\sum_{k=0}^{n-1}\int_{t_k}^{t_{k+1}} \widetilde{\Theta}_k^T(s)\Phi_k(s)\Phi_k^T(s)\widetilde{\Theta}_k(s)ds\notag\\&+\int_{t_n}^t \widetilde{\Theta}_n^T(s)\Phi_n(s)\Phi_n^T(s)\widetilde{\Theta}_n(s)ds=O(\log R(t)),\notag\\
		2)\ &\widetilde{\Theta}^T_{n}(t)P_{n}^{-1}(t)\widetilde{\Theta}_{n}(t)= O(\log R(t)),\notag
	\end{align}
	where $R(t)$ is defined by \eqref{jl}. 	
\end{theorem}
\begin{proof}
Substituting \eqref{ytw} into \eqref{sf1}, we have for $\tau\in[t_k,t_{k+1})$
\begin{align}\label{aaa}
	d\widetilde{\Theta}_k(\tau)=&d(\Theta-\Theta_k(\tau))=-d\Theta_k(\tau)\notag\\
	=&-P_k(\tau)\Phi_k(\tau)d(S)\big\{dY(\tau)-\Phi_k^T(\tau)\Theta_k(\tau)d\tau\big\}\notag\\
	=&-P_k(\tau)\Phi_k(\tau)d(S)\big\{dV(\tau)+\notag\\&\hskip 2cm\Phi_k^T(\tau)\widetilde{\Theta}_k(\tau)d\tau+\widetilde{\Phi}_k^T(\tau)\Theta d\tau\big\}\notag\\
	=&-P_k(\tau)\Phi_k(\tau)\big\{dW(\tau)-d(S)d\widetilde{V}_k(\tau)\big\},
\end{align} where the equation \eqref{aa} is used. Furthermore, by \eqref{ytw1}, \eqref{h} and the definition of $f_k(\tau)$ and $g_k(\tau)$ in Lemma \ref{l2}, we have
\begin{align}\label{zxh}
	d\widetilde{\Theta}_k(\tau)
	=&-P_k(\tau)\Phi_k(\tau)(dW(\tau)+\Phi_k^T(\tau)\widetilde{\Theta}_k(\tau)d\tau\notag\\&+c(S)d\widetilde V_k(\tau)-d(S)d\widetilde V_k(\tau))\\
	=&-P_k(\tau)\Phi_k(\tau)\Big(dW(\tau)+f_k(\tau)d\tau+\frac{1}{2}g_k(\tau)d\tau\Big).\notag
\end{align}
Using Ito formula in Lemma \ref{lA.1} and (\ref{sf1}), we can derive the following equation,
\begin{align}\label{A}
	d\big(&\widetilde{\Theta}_k^T(\tau)P_k^{-1}(\tau)\widetilde{\Theta}_k(\tau)\big)\notag\\	=&\widetilde{\Theta}_k^T(\tau)dP_k^{-1}(\tau)\widetilde{\Theta}_k(\tau)
	+2\widetilde{\Theta}_k^T(\tau)P_k^{-1}(\tau)d\widetilde{\Theta}_k(\tau)\notag\\&+tr(P_k^{-1}(\tau)P_k(\tau)\Phi_k(\tau)\Phi_k^T(\tau)P_k(\tau))d\tau\notag\\	=&\widetilde{\Theta}_k^T(\tau)dP_k^{-1}(\tau)\widetilde{\Theta}_k(\tau)
	\notag\\&-2\widetilde{\Theta}_k^T(\tau)P_k^{-1}(\tau)P_k(\tau)\Phi_k(\tau)(dW(\tau)+f_k(\tau)d\tau\notag\\&+\frac{1}{2}g_k(\tau)d\tau)-tr(P_k^{-1}(\tau)dP_k(\tau)).
\end{align}
According to $P_k^{-1}(\tau)P_k(\tau)=I_{(p+q+r)N}$, we have
\begin{align}\label{DP}dP_k^{-1}(\tau)=-P_k^{-1}(\tau)dP_k(\tau)P_k^{-1}(\tau).\end{align} By this and (\ref{sf1}), the first term on the right hand side (RHS) of (\ref{A}) satisfies
\begin{align}\label{dpt}
	&\widetilde{\Theta}_k^T(\tau)dP_k^{-1}(\tau)\widetilde{\Theta}_k(\tau)\notag\\&=-\widetilde{\Theta}_k^T(\tau)P_k^{-1}(\tau)dP_k(\tau)P_k^{-1}(\tau)\widetilde{\Theta}_k(\tau)\nonumber\\
	&=\widetilde{\Theta}_k^T(\tau)\Phi_k(\tau)\Phi_k^T(\tau)\widetilde{\Theta}_k(\tau)d\tau.
\end{align}
Substituting (\ref{dpt}) into (\ref{A}) yields
\begin{align}\label{dtheta}
	&d\big(\widetilde{\Theta}_k^T(\tau)P_k^{-1}(\tau)\widetilde{\Theta}_k(\tau)\big)\notag\\=&\widetilde{\Theta}_k^T(\tau) \Phi_k(\tau)\Phi_k^T(\tau)\widetilde{\Theta}_k(\tau)d\tau\notag\\&-2\widetilde{\Theta}_k^T(\tau)P_k^{-1}(\tau)P_k(\tau)\Phi_k(\tau)(dW(\tau)\\
	&+f_k(\tau)d\tau+\frac{1}{2}g_k(\tau)d\tau)-tr(P_k^{-1}(\tau)dP_k(\tau))\notag\\
	=&-2g_k^T(\tau)dW(\tau)-2g_k^T(\tau)f_k(\tau)d\tau-tr(P_k^{-1}(\tau)dP_k(\tau)).\notag
\end{align}
Integrating both sides of \eqref{dtheta} on $[t_k,t_{k+1})$, we have
\begin{align}\label{AAAA}
	& \widetilde{\Theta}_k^T(t_{k+1})P_k^{-1}(t_{k+1})\widetilde{\Theta}_k(t_{k+1})\notag \\	=&\widetilde{\Theta}_k^T(t_k)P_k^{-1}(t_k)\widetilde{\Theta}_k(t_k)-2\int_{t_k}^{t_{k+1}}g_k^T(s)dW(s)\notag\\&-2\int_{t_k}^{t_{k+1}}{g_k^T(s)\big(f_k(s)-m_1 g_k(s)\big)ds}\\&{-
		2m_1\int_{t_k}^{t_{k+1}}\|g_k(s)\|^2ds}-\int_{t_k}^{t_{k+1}}tr\big({P^{-1}_k(s)}dP_k(s)\big),\notag
\end{align} where $m_1$ is a positive constant.

Note that by  \eqref{rh}, we have
\begin{align}
	&P_{k+1}(t_{k+1})\mathscr{A}P_k^{-1}(t_{k+1})\Theta\notag\\
	=&\rm{col}\Bigg\{\it{P_{k+\rm{1},\rm{1}}(t_{k+\rm{1}})}\sum_{j \in N_{\rm{1}}}a_{\rm{1}\it{j}}\it{P_{k,j}^{-\rm{1}}}(t_{k+\rm{1}})\theta,\cdots,\notag\\&\it{P_{k+\rm{1},\it N}(t_{k+\rm{1}})}\sum_{j \in N_N}a_{N\it{j}}\it{P_{k,j}^{-\rm{1}}}(t_{k+\rm{1}})\theta\Bigg\}={\rm\Theta}.\notag
\end{align}
Thus, by \eqref{sf} we can get the following equation,
\begin{align}
	\widetilde{\Theta}_{k+1}(t_{k+1})=
	P_{k+1}(t_{k+1})\mathscr{A}P_k^{-1}(t_{k+1})\widetilde{\Theta}_k(t_{k+1}).
\end{align}
Then
\begin{align}
	&\widetilde{\Theta}_{k+1}^T(t_{k+1})P_{k+1}^{-1}(t_{k+1})\widetilde{\Theta}_{k+1}(t_{k+1})\notag\\
	=&\widetilde{\Theta}_k^T(t_{k+1})P_k^{-1}(t_{k+1})
	\mathscr{A}P_{k+1}(t_{k+1})P_{k+1}^{-1}(t_{k+1})\notag\\&P_{k+1}(t_{k+1})\mathscr{A}P_k^{-1}(t_{k+1})\widetilde{\Theta}_k(t_{k+1})\notag\\
	\leq& \widetilde{\Theta}_k^T(t_{k+1})P_k^{-1}(t_{k+1})\widetilde{\Theta}_k(t_{k+1}),
	\label{zuixian}
\end{align} where Lemma \ref{l4} is used in the inequality.
By this inequality and (\ref{AAAA}), we have for all $k$,
\begin{align}\label{DDDD}
	& \widetilde{\Theta}_{k+1}^T(t_{k+1})P_{k+1}^{-1}(t_{k+1})\widetilde{\Theta}_{k+1}(t_{k+1})\notag\\\leq& \widetilde{\Theta}_k^T(t_{k+1})P_k^{-1}(t_{n+1})\widetilde{\Theta}_k(t_{k+1}) \notag\\
	\leq& \widetilde{\Theta}_k^T(t_{k})P_k^{-1}(t_k)\widetilde{\Theta}_k(t_k)-2\int_{t_k}^{t_{k+1}}g_k^T(s)dW(s)\notag\\&-2\int_{t_k}^{t_{k+1}}g_k^T(s)
	\cdot\big(f_k(s)-m_1 g_k(s)\big)ds\\&-2m_1\int_{t_k}^{t_{k+1}}\|g_k(s)\|^2ds-\int_{t_k}^{t_{k+1}}tr\big(P_k^{-1}(s)dP_k(s)\big).\notag
\end{align}
Summing both sides of the above inequality from $0$ to $n-1$, and integrating both sides of \eqref{dtheta} on $[t_n,t)$, we can derive that
\begin{align}\label{D}
	& \widetilde{\Theta}_{n}^T(t)P_{n}^{-1}(t)\widetilde{\Theta}_{n}(t)\notag\\\leq& 	\widetilde{\Theta}_0^T(t_{0})P_{0}^{-1}(t_0)\widetilde{\Theta}_0(t_0)\notag \\&-2\sum_{k=0}^{n-1}\int_{t_k}^{t_{k+1}}g_k^T(s)dW(s)-2\int_{t_n}^{t}g_n^T(s)dW(s)\notag\\
	&-2\sum_{k=0}^{n-1}\int_{t_k}^{t_{k+1}}{g_k^T(s)\big(f_k(s)-m_1 g_k(s)\big)ds}\notag\\
	&-2\int_{t_n}^{t}{g_n^T(s)\big(f_n(s)-m_1 g_n(s)\big)ds}\\
	&-2m_1\sum_{k=0}^{n-1}\int_{t_k}^{t_{k+1}}\|g_k(s)\|^2ds-2m_1\int_{t_n}^{t}\|g_n(s)\|^2ds\notag\\
	&-\sum_{k=0}^{n-1}\int_{t_k}^{t_{k+1}}tr\big(P_k^{-1}(s)dP_k(s)\big)-\int_{t_n}^{t}tr\big(P_n^{-1}(s)dP_n(s)\big).\notag
\end{align}

By applying Lemma \ref{lA.2} to the third and fourth terms on the RHS of \eqref{D}, we have for small $\eta >0$
\begin{align}
	&2\sum_{k=0}^{n-1}\int_{t_k}^{t_{k+1}}g_k^T(s)dW(s)+2\int_{t_n}^{t}g_n^T(s)dW(s)\notag\\
	=&O(1)+o\Big(\Big\{\sum_{k=0}^{n-1}\int_{t_k}^{t_{k+1}}\|g_k(s)\|^2ds+\int_{t_n}^{t}\|g_n(s)\|^2ds\Big\}^{\frac{1}{2}+\eta}\Big).\notag
\end{align}
Substituting the above inequality and (\ref{STR}) into \eqref{D}, we can derive that
\begin{align}\label{1}
	& \widetilde{\Theta}_{n}^T(t)P_{n}^{-1}(t)\widetilde{\Theta}_{n}(t)\notag\\&+2m_1\left(\sum_{k=0}^{n-1} \int_{t_k}^{t_{k+1}}\|g_k(s)\|^2ds+\int_{t_n}^{t}\|g_n(s)\|^2ds\right)\notag\\
	=&O(1)-\sum_{k=0}^{n-1}\int_{t_k}^{t_{k+1}}tr\big(P_k^{-1}(s)dP_k(s)\big)\notag\\&-\int_{t_n}^{t}tr\big(P_n^{-1}(s)dP_n(s)\big).
\end{align}	

Now, we are in a position to estimate the last two terms on the RHS of \eqref{1}. By Lemma \ref{l4}, we have
\beq\label{2}
&-\sum_{k=0}^{n-1} \int_{t_k}^{t_{k+1}}tr(P_k^{-1}(s)dP_k(s))-\int_{t_n}^{t}tr(P_n^{-1}(s)dP_n(s))\\
=&-\sum_{k=0}^{n-1} \int_{t_k}^{t_{k+1}}\frac{d\big|P_{k}(s)\big|}{ \big|P_{k}(s)\big|}-\int_{t_n}^{t}\frac{d\big|P_{n}(s)\big|}{ \big|P_{n}(s)\big|}\\
=&\sum_{k=0}^{n-1} \Big\{\log\big|P_k^{-1}(t_{k+1})\big|-\log\big|P_k^{-1}(t_k)\big|\Big\}\\
&+\log\big|P_n^{-1}(t)\big|-\log\big|P_n^{-1}(t_n)\big|\\
\leq& \sum_{k=0}^{n-1} \Big\{\log\big|P_{k+1}^{-1}(t_{k+1})\big|-\log\big|P_k^{-1}(t_k)\big|\Big\} \\
&+\log\big|P_n^{-1}(t)\big|-\log\big|P_n^{-1}(t_n)\big|\\
=&\log\big|P_{n}^{-1}(t)\big|-\log\big|P_0^{-1}(t_0)\big|\\
\leq& (p+q+r)N\Big\{\log\max_{1\leq i \leq N}\lambda_{\max}\big\{P_{n,i}^{-1}(t)\big\}\Big\}\notag\\&-\log\big|P_0^{-1}(t_0)\big|.
\eeq
By \eqref{sf1} and (\ref{DP}), we can get that
\begin{align}\label{n}
	dP_k^{-1}(\tau)=\Phi_k(\tau)\Phi_k^T(\tau)d\tau.
\end{align}
Integrating both sides of \eqref{n} on $[t_k,t_{k+1})$, we have
\begin{align}\label{nn}
	P_{k}^{-1}(t_{k+1})=P_k^{-1}(t_k)+\int_{t_k}^{t_{k+1}}\Phi_k(s)\Phi_k^T(s)ds.
\end{align}
By the definition of $\{P_{n,i}^{-1}(t)\}$ and the property of the weighted adjacency matrix $\mathcal{A}$, we have
\begin{align}
	&\max_{1\leq i\leq N}\lambda_{\max}\Big\{P_{n,i}^{-1}(t)\Big\}\notag\\
	=& \max_{1\leq i\leq N}\lambda_{\max}\Big\{P_{n,i}^{-1}(t_n)+\int_{t_n}^{t}\phi_{n,i}(s)\{\phi_{n,i}(s)\}^Tds\Big\} \notag\\
	=& \max_{1\leq i\leq N}\lambda_{\max}\Bigg\{ \sum_{j=1}^N a_{ij}^{(n)}P_{0,j}^{-1}(t_0)\notag\\&+\sum_{j=1}^N\sum_{k=0}^{n-1} a_{ij}^{(n-k)} \int_{t_k}^{t_{k+1}}\phi_{k,j}(s)\phi_{k,j}^T(s)ds\notag\\&+\int_{t_n}^{t}\phi_{n,i}(s)\{\phi_{n,i}(s)\}^Tds\Bigg\} \notag\\
	\leq& \max_{1\leq i\leq N}\lambda_{\max}\Big\{ P_{0,i}^{-1}(t_0)\Big\}\notag\\&+\sum_{j=1}^N\sum_{k=0}^{n-1} a_{ij}^{(n-k)} \int_{t_k}^{t_{k+1}}\|\phi_{k,j}(s)\|^2ds \notag\\&+\max_{1\leq i\leq N}\lambda_{\max}\Bigg\{\int_{t_n}^{t}\phi_{n,i}(s)\{\phi_{n,i}(s)\}^Tds\Bigg\}\notag\\
	\leq& \lambda_{\max}\Big\{ P_0^{-1}(t_0)\Big\}+\sum_{j=1}^N\sum_{k=0}^{n}\int_{t_k}^{t_{k+1}}\|\phi_{k,j}(s)\|^2ds.\notag
\end{align}
where $a_{ij}^{(m)}\geq 0$ represent the $(i,j)$th entry of the matrix $\mathcal{A}^m$ satisfying $\sum_{j=1}^N a_{ij}^{(m)}=1$.
Substituting the above  inequality and
into \eqref{1}, we have
\begin{align}
	&\widetilde{\Theta}_{n}^T(t)P_{n}^{-1}(t)\widetilde{\Theta}_{n}(t)\notag\\&+2m_1\left(\sum_{k=0}^{n-1} \int_{t_k}^{t_{k+1}}\|g_k(s)\|^2ds+\int_{t_n}^{t}\|g_n(s)\|^2ds\right)\notag\\
	=&O(\log R(t)).\notag
\end{align}
From this, we can conclude that Theorem \ref{th1} holds.
\end{proof}

Based on the above Theorem \ref{th1}, we can give the following theorem concerning the upper bound of the estimation error of Algorithm \ref{algorithm1}.
\begin{theorem}\label{th2}
	Under Assumptions \ref{c1} and \ref{c2}, the distributed LS algorithm defined by \eqref{sf1}-\eqref{sf} satisfies the following relationship as $t\to+\infty$,
	\begin{center}$\left\|\widetilde{\Theta}_{n}(t)\right\|^2=O\Big( \frac{\log\ R(t)}{\lambda_{\min}^n}\Big)$ a.s.,\end{center}
	where $R(t)$ and $\lambda_{\min}^n$ are defined in \eqref{jl} and \eqref{lambdamin} respectively.
\end{theorem}
\begin{proof} By Theorem \ref{th1}, we can get the following inequality almost surely,
\begin{align}\label{hkz}
	\widetilde{\Theta}^T_{n}(t)\widetilde{\Theta}_{n}(t)\leq& \frac{\widetilde{\Theta}_{n}^T(t)P_{n}^{-1}(t)\widetilde{\Theta}_{n}(t)}{\lambda_{\min}\big\{P_{n}^{-1}(t)\big\}}\notag\\=&\frac{O(\log R(t))}{\lambda_{\min}\big\{P_{n}^{-1}(t)\big\}}.
\end{align}
By lemma\ 8.1.2 in \cite{[30]}, we know that if the graph $\mathcal{G}$ is connected, then $a_{ij}^{(m)}\geq a_{\min}>0$ holds for any $m>D_{\mathcal{G}}$, where $a_{\min}=\min \limits_{i,j\in \mathcal V} a_{ij}^{(D_{\mathcal{G}})}>0$.

In the following, we estimate the smallest eigenvalue of the matrix $P_{n}^{-1}(t)$. By \eqref{sf} and \eqref{n}, we can get that
\begin{align}
	&{\vecc}\{P_{n}^{-1}(t)\}\notag\\
	=&{\vecc} \{P_{n}^{-1}(t_{n})\}+{\vecc} \Big\{\int_{t_n}^{t} \Phi_n(s) \Phi_n^T(s) ds\Big\} \notag\\
	=&\mathscr{A}^{n}{\vecc} \{P_0^{-1}(t_0)\}+{\vecc} \Big\{\int_{t_n}^{t} \Phi_n(s) \Phi_n^T(s) ds\Big\}\notag\\
	&+\sum_{k=0}^{n-1} \mathscr{A}^{n-k}{\vecc} \Big\{\int_{t_k}^{t_{k+1}} \Phi_k(s) \Phi_k^T(s) ds\Big\}\notag.
\end{align}
Thus,
\begin{align}
	&\lambda_{\min}\Big\{P_{n}^{-1}(t)\Big\} =\min_{1\leq i\leq N}\lambda_{\min}\Big\{P_{n,i}^{-1}(t)\Big\} \notag\\
	\geq& \min_{1\leq i\leq N}\lambda_{\min}\Big\{ \sum_{j=1}^N a_{ij}^{(n)}P_{0,j}^{-1}(t_0)\notag\\
	&+\sum_{j=1}^N\sum_{k=0}^{n-D_{\mathcal{G}}} a_{ij}^{(n-k)} \int_{t_k}^{t_{k+1}}\phi_{k,j}(s)\phi_{k,j}^T(s)ds\Big\} \notag
\end{align}
\begin{align}
	\geq& \lambda_{\min}\Big\{ a_{\min}\sum_{j=1}^N P_{0,j}^{-1}(t_0)\notag\\
	&+a_{\min}\sum_{j=1}^N \sum_{k=0}^{n-D_{\mathcal{G}}}  \int_{t_k}^{t_{k+1}}\phi_{k,j}(s)\phi_{k,j}^T(s)ds\Big\}, \notag
\end{align} where $a_{\min}=\min \limits_{i,j\in \mathcal V} a_{ij}^{(D_{\mathcal{G}})}>0$. From this and \eqref{hkz}, we have
\begin{align}\label{v}
	\|\widetilde{\Theta}_{n}(t)\|^2=O\bigg(\frac{ \log R(t)}{\lambda_{\min}^n}\bigg)\ \ \ a.s..
\end{align}
We complete the proof of the theorem.
\end{proof}
\begin{remark} Theorem \ref{th2} gives the upper bound for the estimate $\Theta_{n}(t)$ obtained from Algorithm \ref{algorithm1}. Furthermore, if the cooperative excitation condition (i.e., Assumption \ref{c3}) holds, then the estimate $\Theta_{n}(t)$ converges to the true parameter vector $\Theta$ almost surely, i.e., $\lim\limits_{t\to+\infty}\|\widetilde{\Theta}_{n}(t)\|^2=0$.
\end{remark}

\section{Simulation Results }\label{SIMULATION}

In this section, we provide simulation examples to illustrate the cooperative effect of multiple sensors in Algorithm 1.
 
 We first consider a practical application in which a network of six sensors cooperatively estimates the unknown parameters L, R, and C of an RLC circuit system. The topology structure of sensor network is shown in Fig.\ref{figtp}, and its information exchange weights are determined by the Metropolis rule \citep{xiao2005scheme}. The dynamics of the RLC circuit system are described by the following differential equation:
 \beq\label{RLC}
 L\ddot{\epsilon}_i(t)+R\dot{\epsilon}_i(t)+\frac{1}{C}\epsilon_i(t)=u_i(t),i=1,2,\cdots,6,
 \eeq
 where the unknown parameters $L$, $R$ and $C$ represent the inductance, resistance and capacitance of the RLC system, $u_i(t)$ and $\epsilon_i(t)$ are the input and output measured by sensor $i$, respectively. Clearly, we can rewrite \eqref{RLC} into the following form by denoting $y_i(t)=\dot{\epsilon}_i(t)$,
 \beqn
 \dot{y}_i(t)=\theta^T\phi_i^0(t), i=1,2,\cdots,6,
 \eeqn
 where\\ $\theta=(-\frac{R}{L},-\frac{1}{LC},\frac{1}{L})^T$ and $\phi_i^0(t)=(y_i(t),Sy_i(t),u_i(t))^T.$ 
 
 In simulation, set the true values are $R=3$, $L=5$, $C=5$. The input signals of all sensors are  $u_1(t)=2\cos(t),u_2(t)=\sin(2t),u_3(t)=3\cos(2t),u_4(t)=2\cos(0.5t)+\sin(0.5t),u_5(t)=\sin(2t)+3\cos(2t),u_6(t)=5\cos(0.5t)$, respectively.
  We next apply the distributed LS algorithm proposed in this paper (Algorithm \ref{algorithm1}) with the fusion time interval 0.2, the standard LS algorithm in \cite{2006Continuous},  and the cooperative gradient algorithm \{\cite{6578135,javed2021excitation}\} to estimate the unknown parameter vector $\theta$, respectively. 
 Begin with the same initial states, we conduct the simulation for $100$ runs. 
 
 As can be seen in Fig.\ref{fig_circuit}, the estimate $\theta_{k,i}(t)=(\theta_{k,i}(t;1),\theta_{k,i}(t;2),\theta_{k,i}(t;3))^T$ of each sensor $i$ generated by Algorithm \ref{algorithm1} can converge to its true value since the regressors $\{\phi_i^0(t),i=1.\cdots,6\}$ of these six sensors can jointly satisfy the  Cooperative Excitation Conidtion (Assumption \ref{c3}). However, we can easily verify that
 none of them can satisfy excitation condition \eqref{jd}, as shown in Fig.\ref{fig_circuit},
 the estimates $\theta_{k,i}(t),i=1,2,\cdots,6$ generated by the standard LS algorithm in \cite{2006Continuous} can not converge to the true value.
  
  Fig.\ref{figure_sgbi_1} compares the performance of the cooperative gradient algorithm and Algorithm \ref{algorithm1} and it can be seen that Algorithm \ref{algorithm1} proposed in this paper has better estimation performance.
 
We next consider the stachastic regressors case. Suppose a sensor network consisting of $N=12$ sensors whose dynamics obey the equation \eqref{xt}. The noise process $\{v_{i}(t),i=1,2,\cdots,12,\ t\geq0\}$ in \eqref{xt} are independent scalar Wiener process. The $10$-dimensional unknown parameter vector $\theta$ is
	$\theta=(1,2,\cdots,10)^T$, and the regression vectors  $\phi_i^0(t)\ (i=1,2,\cdots,12)$ are generated by the following method,
	\begin{align}
		x_i(t,q)=&\left\{\begin{array}{ll}
			0.3t+\xi_i(t) \quad &{\rm if}\ \mod(i,3)=0\ \\
			t+\xi_i(t)   \quad &{\rm if}\ \mod(i,3)=1\\
			t^2+\xi_i(t)  \quad &{\rm if}\ \mod(i,3)=2
		\end{array},1\leq q\leq 10\right.
		\notag\\
		\phi_i^0(t)=&[{\textbf 0},\cdots,{\textbf 0},{\textbf e_\tau},{\textbf 0},\cdots,{\textbf 0}
		]^T x_i(t),	\notag
	\end{align}
	where $x_i(t)=(x_i(t,1),\cdots,x_i(t,10))^T$,  ${\textbf 0}$ denotes the $10$-dimensional column vector whose elements are all $0$, $\tau=\mod(i,10)$ and ${\textbf e_\tau}$  represents the $\tau$-th column of the identity matrix $I_{10}$. Let $\{\xi_i(t),i=1,\cdots,12,t\geq 0\}$ be independent scalar Wiener process.
	The topology structure of network graph $\mathcal G$ is shown in Fig.\ref{figtp1}, and its weights are also determined by the Metropolis rule.
	
	The regression vectors $\phi_{k,i}(t)\ (i=1,2,\cdots,12)$ are equal to $\phi_i^0(t)\ (i=1,2,\cdots,12)$ above. Clearly, the regression vectors $\phi_{k,i}(t)\ (i=1,2,\cdots,12)$ can cooperatively satisfy the Cooperative Excitation Condition (Assumption \ref{c3}), but none of them can satisfy the excitation condition \eqref{jd} since there are many zero elements in $\phi_{k,i}(t)$. We conduct the simulation for 100 runs with the same initial states.

	The unknown parameter vector $\theta$ is estimated by Algorithm \ref{algorithm1} with the fusion time interval 0.2 and the standard LS algorithm in \cite{2006Continuous}. The following Fig.\ref{fig2} shows the change of the mean square errors (MSEs) (averaged over 100 runs) of the $12$ sensors with time $t$ for these two algorithms.  In the upper part of Fig.\ref{fig2} where the MSEs for Algorithm \ref{algorithm1} proposed in this paper can converge to zero in a certain rate since the sensors can jointly satisfy the Cooperative Excitation Conidtion (Assumption \ref{c3}). While in the lower part of Fig.\ref{fig2}, the estimates are obtained from the standard LS algorithm without information exchange between sensors, the MSEs of the sensors can not converge to zero since the sensors do not have enough excitation.
	
	From the above simulation examples, the joint effect of the multi-sensor network can be revealed in a sense that the sensors can accomplish the estimation tasks through information exchange even if none of the sensors can not.
	\begin{figure}[htbp]
		\centerline{\includegraphics[width=0.6\columnwidth]{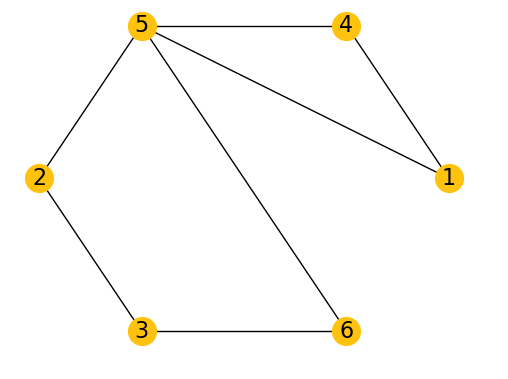}}
		\caption{The topology structure of sensor network.}
		\label{figtp}
	\end{figure}
\begin{figure}[htbp]
	\centerline{\includegraphics[width=1.2\columnwidth]{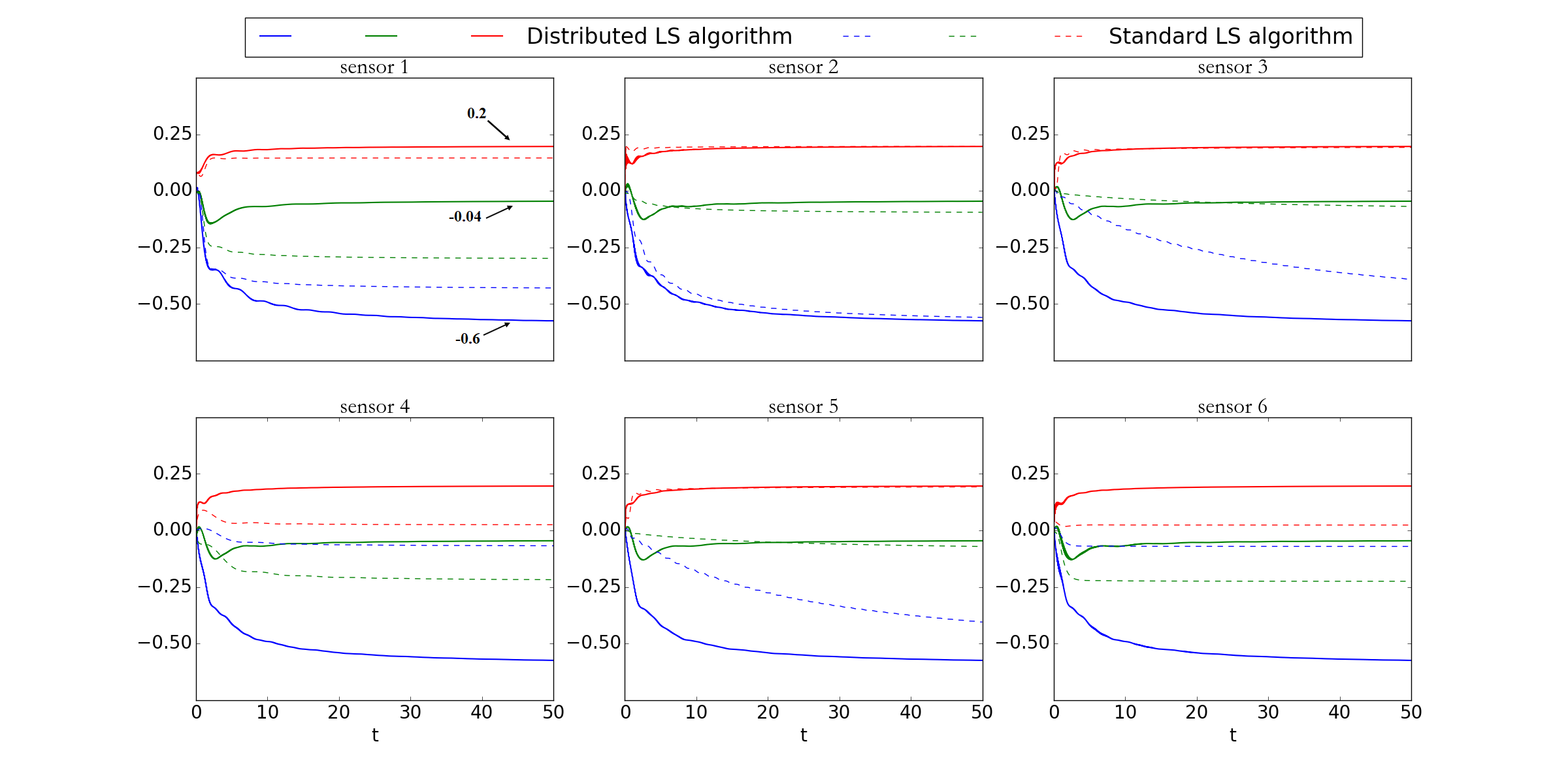}}
	\caption{The estimate $\theta_{k,i}(t)$ of each sensor $i$ ($i=1,2,\cdots,6$) obtained from Algorithm \ref{algorithm1} and standard LS algorithm, respectively.}
	\label{fig_circuit}
\end{figure}

\begin{figure}[htbp]
	\centerline{\includegraphics[width=1.2\columnwidth]{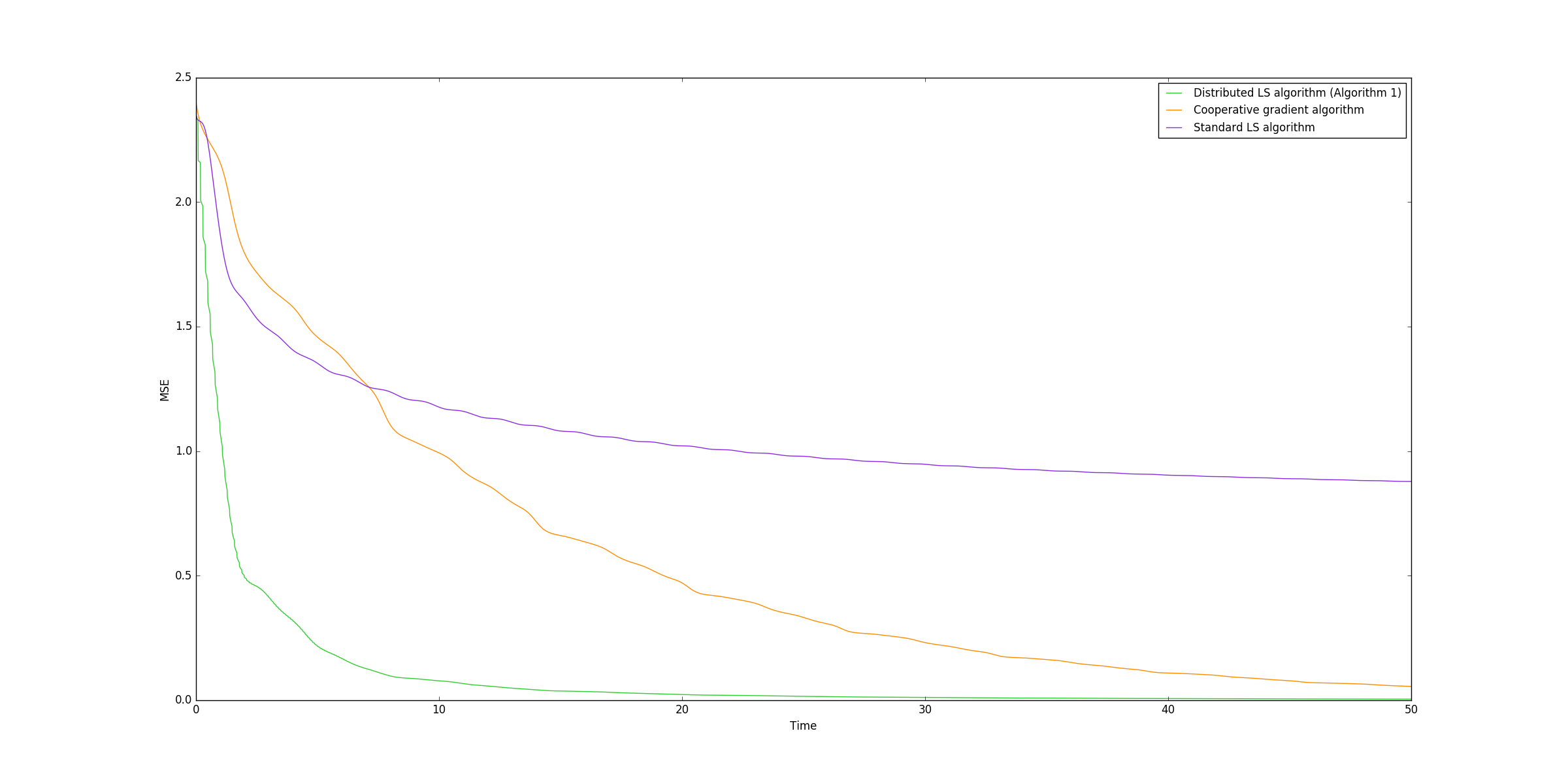}}
	\caption{The comparison of several algorithms.}
	\label{figure_sgbi_1}
\end{figure}

\begin{figure}[htbp]
	\centerline{\includegraphics[width=0.6\columnwidth]{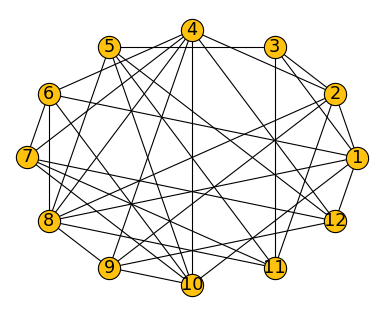}}
	\caption{The topology structure of network graph $\mathcal G$.}
	\label{figtp1}
\end{figure}
	\begin{figure}[htbp]
		\centerline{\includegraphics[width=0.95\columnwidth]{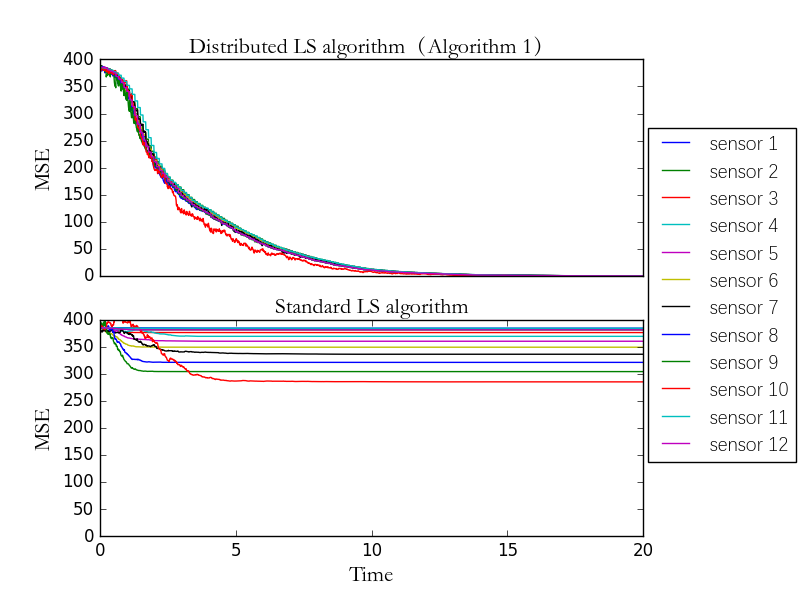}}
		\caption{The MSEs of sensors where the estimates are obtained from Algorithm \ref{algorithm1} and the standard LS algorithm, respectively.}
		\label{fig2}
	\end{figure}

\section{Concluding Remarks}\label{conclusion}
This paper mainly considers the distributed parameter estimation problem of continuous-time linear stochastic regression systems over sensor networks. We propose the distributed LS algorithm by using the continuous-time noisy signals to estimate the unknown parameter vector. The upper bounds of the estimation error for the proposed algorithm is obtained, and the convergence analysis is further presented under a cooperative excitation condition. Different from most results on the distributed estimation of continuous-time systems in the existing literature, our results are obtained without relying on the boundedness and PE conditions of the regression vectors. Moreover, the cooperative excitation condition can reveal the joint effect of multiple sensors in the proposed distributed LS algorithm. Many interesting problems deserve to be further investigated, such as the distributed estimation problem of the continuous-time nonlinear stochastic regression systems, the performance analysis of the continuous-time distributed adaptive filtering, the distributed adaptive control problem, etc.

\bibliographystyle{apalike}        
\bibliography{my_ref}      

\end{document}